\documentclass[conference, 10pt]{IEEEtran}
\IEEEoverridecommandlockouts
\def\BibTeX{{\rm B\kern-.05em{\sc i\kern-.025em b}\kern-.08em
    T\kern-.1667em\lower.7ex\hbox{E}\kern-.125emX}}
\usepackage{comment, url, cite, color}
\usepackage{caption, subcaption, multirow, multicol}
\usepackage{amsthm, amsmath, amssymb, amsfonts, mathtools, bm}
\allowdisplaybreaks
\theoremstyle{definition}
\newtheorem{theorem}{Theorem}

\newtheorem{proposition}[theorem]{Proposition}
\usepackage[ruled,vlined]{algorithm2e}
\SetKwInOut{Parameters}{Parameters}
\SetKw{Continue}{continue}
\usepackage{setspace}
\newlength\mylen

\DeclareMathOperator{\adjoint}{\mathsf{H}}
\DeclareMathOperator{\GL}{GL}
\DeclareMathOperator{\trace}{Tr}
\DeclareMathOperator{\diag}{diag}
\DeclareMathOperator*{\argmin}{argmin}

\begin{document}
\title{
ISS2: An Extension of Iterative Source Steering Algorithm
for Majorization-Minimization-Based Independent Vector Analysis
}
\author{\IEEEauthorblockN{Rintaro Ikeshita and Tomohiro Nakatani}
\IEEEauthorblockA{\textit{NTT Corporation, Japan}}
}
\maketitle

\begin{abstract}
A majorization-minimization (MM) algorithm for independent vector analysis optimizes a separation matrix 
$W \coloneqq [\bm{w}_1, \ldots, \bm{w}_m]^{\adjoint} \in \mathbb{C}^{m \times m}$
by minimizing a surrogate function of the form
$\mathcal{L}(W) \coloneqq \sum_{i = 1}^m \bm{w}_i^{\adjoint} V_i \bm{w}_i - \log | \det W |^2$,
where
$m \in \mathbb{N}$ is the number of sensors and
positive definite matrices $V_1,\ldots,V_m \in \mathbb{C}^{m \times m}$ are constructed in each MM iteration.
For $m \geq 3$, no algorithm has been found to obtain a global minimum of $\mathcal{L}(W)$.
Instead, block coordinate descent (BCD) methods with closed-form update formulas have been developed for minimizing $\mathcal{L}(W)$
and shown to be effective.
One such BCD is called iterative projection (IP) that updates one or two rows of $W$ in each iteration.
Another BCD is called iterative source steering (ISS) that updates one column of the mixing matrix $A \coloneqq W^{-1}$ in each iteration.
Although the time complexity per iteration of ISS is $m$ times smaller than that of IP,
the conventional ISS converges slower than the current fastest IP (called $\text{IP}_2$) that updates two rows of $W$ in each iteration.
We here extend this ISS to $\text{ISS}_2$ that can update two columns of $A$ in each iteration
while maintaining its small time complexity.
To this end, we provide a unified way for developing new ISS type methods
from which $\text{ISS}_2$ as well as the conventional ISS can be immediately obtained in a systematic manner.
Numerical experiments to separate reverberant speech mixtures show that
our $\text{ISS}_2$ converges in fewer MM iterations than the conventional ISS, and is comparable to $\text{IP}_2$.
\end{abstract}
\begin{IEEEkeywords}
independent component analysis (ICA),
independent vector analysis (IVA),
majorization-minimization (MM),
block coordinate descent (BCD)
\end{IEEEkeywords}

\section{Introduction}
% Context: Set the stage, motivate the general topic
% Gap: Explain your specific problem and why existing work does not adequately solve it
% Innovation: State what you've done that is new, and explain how it helps fill the gap

\IEEEPARstart{I}{ndependent} component analysis (ICA)~\cite{comon2010handbook}
and its extension, independent vector analysis (IVA)~\cite{kim2007},
are fundamental blind source separation (BSS) methods
that have been applied in numerous fields.
Although theoretical properties such as identifiability of ICA~\cite[Chapter 4]{comon2010handbook}\cite{afsari2008sensitivity}
and IVA~\cite{anderson2014independent,lahat2018joint,lahat2016alternative}
have been well studied,
the algorithms developed for them still need improvement
because 
fast and stable optimization is indispensable
when applied to real-world applications.

Early algorithms for ICA include 
Infomax~\cite{bell1995infomax}
and the relative (or natural) gradient method~\cite{cardoso1996equivariant,amari1996natural-gradient}.
%,cichocki1996robust}.
To accelerate these gradient-based algorithms using curvature information,
several second-order algorithms with (relative) Hessian approximation were proposed~\cite{zibulevsky2003blind,palmer2008newton,choi2007relative,ablin2018faster}.
In another research direction, a primal-dual splitting algorithm (e.g.,~\cite{komodakis2015playing}) for IVA~\cite{yatabe2018determined}
and its heuristic extension~\cite{yatabe2021determined} based on the plug-and-play scheme
were recently developed.
However, all the above algorithms rely on good policies for determining hyperparameters such as step size,
and it is usually difficult to find such policies that are suitable for any kind of signals.
Other famous methods, such as FastICA~\cite{hyvarinen1999fastica} and its improvement~\cite{ablin2018faster-O},
assume orthogonal constraint for the separated signals,
which is not necessarily optimal, especially for short signals. 

To avoid these problems,
a majorization-minimization (MM) algorithm~\cite{lange2016mm} for ICA~\cite{ono2010auxica,ablin2019stochastic}
and IVA~\cite{ono2011auxiva} without such tuning parameters as the step size
was proposed about a decade ago (see Section~\ref{sec:MM})
and has been studied extensively (mainly in the audio source separation community)
because it can attain fast and stable optimization.
Interestingly, a majorizer (or a surrogate function) constructed in the MM algorithm
had already been studied in the ICA literature~\cite{pham2001,degerine2007maxdet,yeredor2012SeDJoCo},
not related to the MM approach.

Because the majorizer is non-convex and obtaining a global minimum is difficult,
two families of block coordinate descent (BCD) methods~\cite{nocedal-Jorge2006optimization} with closed-form update formulas were developed.
One is called \textit{iterative projection (IP)}~\cite{ono2018asj,nakashima2021faster,scheibler2020ive,scheibler2021ipa,ike2021ive}
that updates one or \underline{two} rows of the separation matrix $W \in \mathbb{C}^{m \times m}$ in each BCD iteration, where $m \in \mathbb{N}$ is the number of sensors.
The other is called \textit{iterative source steering (ISS)}~\cite{scheibler2020iss}
that updates \underline{one} column of the mixing matrix $A \coloneqq W^{-1}$ in each BCD iteration.
Although ISS reduces the time complexity of IP by a factor of $1 / m$,
it requires more MM iterations to converge than the current fastest IP (called $\text{IP}_2$) that updates \underline{two} rows of $W$ in each iteration.

In this paper,
we extend the conventional ISS so that it can update
\underline{two} columns of $A \coloneqq W^{-1}$ in each iteration
while keeping its small time complexity.
The numerical simulation demonstrates the effectiveness of the proposed approach.

%\subsection{Notation}
\textbf{Notation:}
Let $\GL(m)$ be the set of all $m \times m$ nonsigular matrices over $\mathbb{C}$,
and let $\mathcal{S}_{+}^m \subset \mathbb{C}^{m \times m}$ be the set of all Hermitian positive semidefinite matrices.
For a matrix $A \in \mathbb{C}^{m \times n}$,
let 
$A^\top$ and $A^{\adjoint}$ denote the transpose and conjugate transpose of $A$,
$A_{ij}$ be the $(i,j)$th entry of $A$,
$A_{i, \bullet} \in \mathbb{C}^{1 \times n}$ be the $i$th row of $A$,
$A_{i:i+d, \bullet} \coloneqq [A_{i, \bullet}^\top, \ldots, A_{i + d, \bullet}^\top]^\top \in \mathbb{C}^{(d + 1) \times n}$,
and $\diag(A_{i, \bullet}) \in \mathbb{C}^{n \times n}$ be the diagonal matrix whose diagonal entries are $A_{i, \bullet}$.
The identity and zero matrices are denoted as $I_d \in \mathbb{C}^{d \times d}$ and $O_{i,j} \in \mathbb{C}^{i \times j}$, respectively.

\section{Background}
\subsection{Independent Vector Analysis (IVA)}

Consider a set of $K \geq 1$ linear mixtures:
\begin{align}
X^{[k]} = A^{[k]} S^{[k]} \in \mathbb{C}^{m \times n},
\quad
k = 1,\ldots,K,
\end{align}
where
$m \in \mathbb{N}$ is the number of sensors,
$n \in \mathbb{N}$ is the number of sample points,
$X^{[k]} \in \mathbb{C}^{m \times n}$ is an observation,
$S^{[k]} \in \mathbb{C}^{m \times n}$ is the original $m$ source signals,
and $A^{[k]} \in \GL(m)$ is called a \textit{mixing matrix}.
The goal of IVA is to estimate the set of the \textit{separation matrices} $W^{[k]} \in \GL(m)$, $k = 1,\ldots,K$
satisfying
\begin{align}
\label{eq:ambiguity}
W^{[k]} A^{[k]} = D^{[k]} \Pi,
\quad
k = 1,\ldots,K,
\end{align}
where $D^{[k]}$ and $\Pi$ are respectively the
arbitrary diagonal and permutation matrices of size $m \times m$
that correspond to the \textit{scale and permutation ambiguities} of separated signals $Y^{[k]} \coloneqq W^{[k]} X^{[k]}$.
Note that permutation matrix $\Pi$ must be independent of $k$ to ensure that
the orders of the $K$ separated signals $Y^{[1]},\ldots,Y^{[K]}$ are aligned between different mixtures.

To achieve the above,
IVA relies on the assumption that, for each $i = 1,\ldots,m$ and $j = 1,\ldots,n$,
the vector
\begin{align}
\bm{y}_{ij} \coloneqq [\, Y_{ij}^{[1]}, \ldots, Y_{ij}^{[K]} \,]^\top \in \mathbb{C}^K
\end{align}
follows a probability density function with
second or higher-order correlation~\cite{anderson2014independent,lahat2018joint,lahat2016alternative}.
Also, it is commonly assumed that the random variables $\{ \bm{y}_{ij} \}_{ij}$ are mutually independent.
Under this model, the negative log-likelihood,
which yields a cost function of $\mathcal{W} \coloneqq ( W^{[k]} )_{k = 1}^K$,
is expressed as:
\begin{align}
\nonumber
&\mathcal{L}_0 (\mathcal{W}) 
\coloneqq - \frac{1}{n} \log p(X^{[1]},\ldots,X^{[K]}; \mathcal{W})
\\
\label{eq:risk}
=&\,
- \frac{1}{n} \sum_{i = 1}^m \sum_{j = 1}^n \log p(\bm{y}_{ij})
- \sum_{k = 1}^K \log | \det W^{[k]} |^2.
\end{align}

\subsection{Majorization-Minimization Algorithm for IVA}
\label{sec:MM}

An MM algorithm for ICA was proposed by (Ono and Miyabe, 2010~\cite{ono2010auxica}),
rediscovered by (Ablin, Gramfort, Cardoso, and Bach, 2019~\cite{ablin2019stochastic}), and
extended for IVA by (Ono, 2011~\cite{ono2011auxiva}).
Here we briefly review it.

Let $p(\bm{y})$ be a circularly-symmetric probability density function of a random variable $\bm{y}$,
and
a function $\varphi \colon \mathbb{R}_{\geq 0} \to \mathbb{R}$ be given by
$\varphi( \| \bm{y} \|_2 ) \coloneqq - \log p(\bm{y})$
with $\| \bm{y} \|_2 \coloneqq \sqrt{ \bm{y}^{\adjoint} \bm{y} }$.
We say that $p(\bm{y})$ is \textit{super-Gaussian} if
$\varphi'(r) / r$ is decreasing on $r \in (0, \infty) = \mathbb{R}_{> 0}$,
where $\varphi'$ is the first derivative of $\varphi$
(see, e.g.,~\cite{ono2010auxica,palmer2006variational}, and \cite[pp.\ 60--61]{benveniste1990}).
For instance, a generalized Gaussian distribution (GGD)
\begin{align}
\label{eq:GGD}
\varphi(\| \bm{y} \|_2) = \| \bm{y} \|_2^\beta + \mathrm{const.}, \quad 0 < \beta < 2
\end{align}
is super-Gaussian.
GGD with $\beta = 1$ is nothing but the Laplace distribution.
For a super-Gaussian $\varphi(r)$, we have (see~\cite{ono2010auxica})
\begin{align}
\label{ineq:super-gaussian}
\varphi (r) = \min_{\alpha > 0} \Big[\, \frac{\varphi'(\alpha)}{2\alpha} \cdot r^2 + \Big(
\varphi(\alpha) - \frac{ \alpha \varphi'(\alpha) }{2}
\Big) \,\Big]
\end{align}
for all $r \in \mathbb{R}_{> 0}$
and its minimum is attained at $\alpha = r$.
Using \eqref{ineq:super-gaussian} for each 
$\varphi( \| \bm{y}_{ij} \|_2 ) \coloneqq - \log p( \bm{y}_{ij} )$ in \eqref{eq:risk},
we can develop an MM algorithm for IVA~\cite{ono2011auxiva} that alternately updates
an auxiliary variable $\Lambda \in \mathbb{R}_{\geq 0}^{m \times n}$
and
$\mathcal{W}$ by repeating
\begin{align}
\label{eq:lambda}
\Lambda_{ij} &\leftarrow
\frac{
    \varphi'(\| \bm{y}_{ij} \|_2)
}{ 
    \| \bm{y}_{ij} \|_2 
},
\quad i = 1,\ldots,m;~ j=1,\ldots,n,
\\
\label{problem:W}
W^{[k]} &\in \argmin_{W^{[k]} \in \GL(m)} \mathcal{L}^{[k]} (W^{[k]}, \Lambda),
\quad k = 1,\ldots,K,
\end{align}
where we define
\begin{align}
\label{eq:L}
&\hspace{-8 mm}
\mathcal{L}^{[k]}(W^{[k]}, \Lambda) = \sum_{i = 1}^m
(\bm{w}_i^{[k]})^{\adjoint} V_i^{[k]} \bm{w}_i^{[k]} - \log | \det W^{[k]} |^2,
\\
\label{eq:Vi}
V_i^{[k]} &= \frac{1}{2n} X^{[k]} \diag(\Lambda_{i, \bullet}) \, (X^{[k]})^{\adjoint} \in \mathcal{S}_+^m.
\\
\label{eq:w}
\bm{w}_i^{[k]} &= (W_{i, \bullet}^{[k]})^{\adjoint} 
\quad (\Leftrightarrow W^{[k]} = [\bm{w}_1^{[k]},\ldots,\bm{w}_m^{[k]}]^{\adjoint}),
\end{align}
Note that $\Lambda_{i, \bullet} \in \mathbb{R}_{\geq 0}^{1 \times n}$ in \eqref{eq:Vi} is the $i$th row of $\Lambda$.

When $m = 2$ and $W^{[k]} \in \mathbb{C}^{2 \times 2}$, problem \eqref{problem:W} has a closed-form solution~\cite{ono2012auxiva-stereo,degerine2007maxdet}.
However, for $m \geq 3$, no algorithm has been found that obtains a global minimum of \eqref{problem:W},
and several BCD algorithms were developed.
In this paper, we refer to such MM-based IVA approaches as \textit{MM+BCD}.

Hereafter, for ease of notation,
we omit the upper right index $\cdot^{[k]}$ 
when discussing~\eqref{problem:W}--\eqref{eq:w}
and simply denote the objective function $\mathcal{L}^{[k]}(W^{[k]}, \Lambda)$ as $\mathcal{L}(W)$.

\section{Proposed MM+BCD Algorithm}
\label{sec:ISS}

We generalize the definition of
iterative source steering (\text{ISS}) to be a family of MM+BCD algorithms that update several columns of $A \coloneqq W^{-1}$ in each iteration
based on the minimization of $\mathcal{L}(W)$ with respect to those columns.
The conventional \text{ISS}~\cite{scheibler2020iss} (called $\text{ISS}_{1}$) updates
one column of $A$ in each iteration.
We extend this $\text{ISS}_1$ to $\text{ISS}_2$ so that it can update
two columns of $A$ in each iteration.
To this end, we newly provide a unified way to develop $\text{ISS}_d$ for any $d \geq 1$.

\subsection{Definition of $\text{ISS}_d$}
\label{sec:iss:def}

Let $d$ be a divisor of $m$ and $L \coloneqq m / d$.
Consider the partition of $A$ into $L$ submatrices with $d$ columns:
\begin{align}
A = [\, \underbrace{A^{(1)}}_{d} \mid \cdots \mid \underbrace{A^{(L)}}_{d} \,] \in \mathbb{C}^{m \times m}.
\end{align}
%
%For instance, $A^1$ is the first $d$-columns of $A$.
$\text{ISS}_d$ is an MM+BCD method that cyclically updates
\begin{align}
\Lambda \to (W, A^{(1)}) \to (W, A^{(2)}) \to \cdots \to (W, A^{(L)})
\end{align}
one by one based on \eqref{eq:lambda} for updating $\Lambda$ and
\begin{align}
\label{problem:iss:d}
(W, A^{(\ell)}) \in \argmin_{(W,\, A^{(\ell)})} \, \{ \mathcal{L} (W) \mid W A = I_m \}
\end{align}
for updating $(W, A^{(\ell)})$ with $\ell = 1,\ldots,L$.
When $d = 1$,
our definition of $\text{ISS}_1$ coincides with the conventional $\text{ISS}_1$~\cite{scheibler2020iss}.

\subsection{Multiplicative update (MU) formulation for $\text{ISS}_d$}
\label{sec:ISS:MU}

We show that $\text{ISS}_d$ can be written as a multiplicative update (MU) algorithm for $W$ (or equivalently $Y = WX$).
%which is summarized in Algorithm~\ref{alg:iss:pseudo}.
To begin with, we provide the following proposition.
\begin{proposition}
Update rule~\eqref{problem:iss:d} with $\ell = 1$ is equivalent to the following MU rule for $W$ (and $A$):
\begin{align}
\label{eq:iss:mu:1}
T & \in \argmin_T \left\{ \mathcal{L} (TW) \mid T \in \mathcal{D}_{\text{ISS}_d} \right\},
\\
\label{eq:iss:mu:2}
W &\leftarrow T W \quad(\text{and} ~ A \leftarrow A T^{-1}),
\end{align}
where we define
\begin{align*}
\mathcal{D}_{\text{ISS}_d} \coloneqq
\left\{
\begin{bmatrix}
P & O_{d, m - d} \\
Q & I_{m - d}
\end{bmatrix}    
\;\middle|\;
P \in \GL(d),
\,
Q \in \mathbb{C}^{(m - d) \times d}
\right\}.
\end{align*}
\end{proposition}
\begin{proof}
By the update of $A^{\text{new}} \leftarrow A T^{-1}$ with $T^{-1} \in \mathcal{D}_{\text{ISS}_d}$,
$A^{(1)}$ can take an arbitrary value
while $[A^{(2)},\ldots,A^{(L)}]$ remains unchanged.
To keep the constraint $W A = I_m$ in \eqref{problem:iss:d}, $W$ must be uniquely updated to $W^{\text{new}} \leftarrow TW$:
\begin{align*}
\underbrace{
\left[
\begin{array}{cc}
P^{-1} & O_{d,m-d} \\
- Q P^{-1} & I_{m-d}
\end{array}
\right]
W
}_{W^{\text{new}} \;=\; T W}
\;
\underbrace{
A
\left[
\begin{array}{cc}
P & O_{d,m-d} \\
Q & I_{m-d}
\end{array}
\right]
}_{A^{\text{new}} \;=\; A T^{-1}}
= I_m.
\end{align*}
(Note that the set $\mathcal{D}_{\mathrm{ISS}_d}$ is closed under matrix inversion.)
This $T$ belongs to and runs over $\mathcal{D}_{\text{ISS}_d}$
when $T^{-1}$ runs over $\mathcal{D}_{\text{ISS}_d}$.
%which completes the proof.
Thus, Eq.~\eqref{problem:iss:d} with $\ell = 1$ is equivalent to \eqref{eq:iss:mu:1}--\eqref{eq:iss:mu:2}.
\end{proof}

We next show that Eq.~\eqref{problem:iss:d} with $\ell \in \{2,\ldots,L \}$
can also be rewritten in the same way as \eqref{eq:iss:mu:1}--\eqref{eq:iss:mu:2}
by properly permuting the rows of $W$ and the columns of $A$ in advance.
To see this, let us define a (block) permutation matrix
\begin{align}
\label{eq:permu-mat}
{
%\small
\Pi_d =
\left[\begin{array}{c|ccc}
& I_d & & \\
& & \ddots & \\
& & & I_d \\ \hline
I_d & & &
\end{array}\right]
\in \mathbb{C}^{m \times m}
}
\end{align}
and permute the rows of $(W, Y, \Lambda)$ and columns of $A$ by
\begin{align*}
%\label{eq:permu:row}
W &\leftarrow \Pi_d^{\ell - 1} W,
\quad
Y \leftarrow \Pi_d^{\ell - 1} Y,
\quad
\Lambda \leftarrow \Pi_d^{\ell - 1} \Lambda,
\\
%\label{eq:permu:col}
A &\leftarrow A (\Pi_d^{\ell - 1})^\top = [\,A^{(\ell)}, \ldots, A^{(L)}, A^{(1)},\ldots,A^{(\ell - 1)}\,].
\end{align*}
This permutation keeps both the objective function and constraint $WA = I_m$ in \eqref{problem:iss:d} since $\Pi_d \Pi_d^\top = I_m$.
Also, the first $d$ columns of $A (\Pi_d^{\ell - 1})^\top$ are $A^{(\ell)}$.
Thus, Eq.~\eqref{problem:iss:d} with $\ell \geq 2$
is also essentially equivalent to \eqref{eq:iss:mu:1}--\eqref{eq:iss:mu:2}.
%by replacing $(W, A, Y, \Lambda)$ with \eqref{eq:permu:row}--\eqref{eq:permu:col} in advance.
Due to this observation, we only need to address problem~\eqref{eq:iss:mu:1} below.

\subsection{Derivation of $\text{ISS}_2$ (and new derivation of $\text{ISS}_1$)}
\label{sec:ISS:alg}

We discuss problem \eqref{eq:iss:mu:1} for general $d \geq 1$
and develop a closed-form solution for it when $d = 2$ (proposed $\text{ISS}_2$) and $d = 1$ (new derivation of $\text{ISS}_1$).

For $T \in \mathcal{D}_{\text{ISS}_d}$,
the $i$th row vector of $P$ (resp.\ $Q$) is denoted as $\bm{p}_i^{\adjoint} \in \mathbb{C}^{1 \times d}$ (resp.\ $\bm{q}_{d + i}^{\adjoint} \in \mathbb{C}^{1 \times d}$):
\begin{align}
P &= [\bm{p}_1, \ldots, \bm{p}_d]^{\adjoint} \in \mathbb{C}^{d \times d},
\\
Q &= [\bm{q}_{d + 1}, \ldots, \bm{q}_{m}]^{\adjoint} \in \mathbb{C}^{(m - d) \times d}.
\end{align}
Then the objective function $\mathcal{L}(TW)$ can be expressed as
\begin{align}
\nonumber
\mathcal{L} (T W) =&\;
\sum_{i = 1}^d \bm{p}_i^{\adjoint} G_i \bm{p}_i - \log | \det P |^2
\\
\label{eq:ISS:loss}
&+
\sum_{i = d + 1}^m
\begin{bmatrix}
\bm{q}_i \\
1
\end{bmatrix}^{\adjoint}
\begin{bmatrix}
G_i & \bm{g}_i \\
\bm{g}_i^{\adjoint} & c_i
\end{bmatrix}
\begin{bmatrix}
\bm{q}_i \\
1
\end{bmatrix}
+
\text{const.},
\end{align}
where for each $i = 1,\ldots,m$,
\begin{align}
\nonumber
G_i &= W_{1:d,\bullet} V_i W_{1:d,\bullet}^{\adjoint}
= \frac{1}{2n} Y_{1:d,\bullet} \diag(\Lambda_{i, \bullet}) \, (Y_{1:d,\bullet})^{\adjoint} \in \mathbb{C}^{d \times d},
\\
\nonumber
\bm{g}_i &= W_{1:d,\bullet} V_i W_{i,\bullet}^{\adjoint} 
= \frac{1}{2n} Y_{1:d,\bullet} \diag(\Lambda_{i, \bullet}) \, (Y_{i,\bullet})^{\adjoint} \in \mathbb{C}^{d \times 1},
\end{align}
and $c_i \in \mathbb{C}$ is constant.
Since the variables $P$ and $Q$ are split in \eqref{eq:ISS:loss}, we can optimize them separately.

\subsubsection{Optimization of $Q$ for general $d \geq 1$}

Since the objective function $\mathcal{L}(TW)$ is quadratic with respect to $Q$,
and $G_{d + 1}, \ldots, G_m$ are positive definite in general (if $d \leq n$),
the global optimal solution for $Q$ is obtained as
\begin{align}
\bm{q}_i = - G_i^{-1} \bm{g}_i \in \mathbb{C}^{d \times 1}, \quad i = d + 1, \ldots, m.
\end{align}
Along with this, the separated signals are updated as
\begin{align}
\label{eq:Yi-update}
Y_{i, \bullet} &\leftarrow Y_{i, \bullet} + \bm{q}_i^{\adjoint} Y_{1:d,\bullet} = Y_{i, \bullet} - \bm{g}_i^{\adjoint} G_i^{-1} Y_{1:d,\bullet} \in \mathbb{C}^{1 \times n}
\end{align}
for each $i = d + 1,\ldots,m$.
Note that in $\text{ISS}_d$ we only need to update $(Y, \Lambda)$ but not $W$ since the surrogate function $\mathcal{L} (TW)$
given by \eqref{eq:ISS:loss} can be constructed from $(Y, \Lambda)$ only.

\subsubsection{Optimization of $P$ for $d \geq 3$}

We want to solve
\begin{align}
\label{problem:maxdet}
P \in \argmin_{P \in \GL(d)} \; \sum_{i = 1}^d \bm{p}_i^{\adjoint} G_i \bm{p}_i - \log | \det P |^2.
\end{align}
However, as mentioned in Section~\ref{sec:MM},
obtaining a global minimum of~\eqref{problem:maxdet} for $d \geq 3$ is a long-standing open problem~\cite{degerine2007maxdet},
and we leave this task for future work.

\subsubsection{Optimization of $P$ for $d = 2$ ($\text{ISS}_2$ case)}
\label{sec:ISS2}

When $d = 2$,
problem~\eqref{problem:maxdet} is known to have a closed-form solution~\cite{ono2012auxiva-stereo}:
\begin{align}
\label{eq:ip2:1}
H &= G_1^{-1} G_2 \in \mathbb{C}^{2 \times 2},
\\
\nonumber
\theta_1 &= \frac{\trace(H) + \sqrt{ (\trace(H))^2 - 4 \, \operatorname{det}(H) }}{2},
\quad
\theta_2 = \frac{\det H}{\theta_1},
\\
\bm{u}_1 &= \begin{bmatrix}
    H_{22} - \theta_1 \\
    - H_{21}
\end{bmatrix},
\quad
\bm{u}_2 = \begin{bmatrix}
    - H_{12} \\
    H_{11} - \theta_2
\end{bmatrix} \in \mathbb{C}^{2 \times 1},
\\
\label{eq:ip2:2}
\bm{p}_i &= \frac{\bm{u}_i}{(\bm{u}_i^{\adjoint} G_i \bm{u}_i)^{\frac{1}{2}}} \in \mathbb{C}^{2 \times 1}, \quad i = 1,2.
\end{align}
Along with this,
the separated signals are updated by $Y_{1:2,\bullet} \leftarrow P Y_{1:2,\bullet}$.
The proposed $\text{ISS}_2$ is summarized in Algorithm~\ref{alg:main}.

\subsubsection{Optimization of $P$ for $d = 1$ ($\text{ISS}_1$ case)}
\label{sec:ISS1}

When $d = 1$, $p_1 = G_1^{-\frac{1}{2}}$ gives a global minimum of \eqref{problem:maxdet}.
The obtained $\text{ISS}_1$ is identical to the conventional $\text{ISS}_1$~\cite{scheibler2020iss}.
Our new derivation has an advantage of providing a systematic way to discuss $\text{ISS}_d$,
which enabled us to generalize $\text{ISS}_1$ to $\text{ISS}_2$ as above.

\begin{algorithm}[t]
\caption{IVA by $\text{ISS}_2$}
\label{alg:main}
{
\setstretch{1.1}
\SetNoFillComment
\DontPrintSemicolon
\KwIn{$X^{[k]} \in \mathbb{C}^{m \times n}$ $(k = 1,\ldots,K)$}
\nl Initialize $W^{[k]}$ as a whitening matrix for $k = 1,\ldots,K$.\;
\nl $Y^{[k]} \leftarrow W^{[k]} X^{[k]}$ for each $k = 1,\ldots,K$.\;
\nl\Repeat (\tcp*[f]{outer MM loop}) {\rm some convergence criterion is met}{
    \nl $\Lambda_{ij} \leftarrow \varphi'(\| \bm{y}_{ij} \|_2 + \varepsilon) \,/\, (\| \bm{y}_{ij} \|_2 + \varepsilon)$, where $\varepsilon = 10^{-10}$ is added to improve numerical stability.\;
    %\nl $\Lambda_{ij} \leftarrow \Lambda_{ij} + 10^{-10}$ for numerical stability\;
    \nl\For (\tcp*[f]{inner BCD loop}) {\rm $\ell = 1, \ldots, \frac{m}{2}$}{
        \nl\For{$k = 1, \ldots, K$}{
            \nl\tcc{\textcolor{black}{Update $Y_{3:m,\bullet}^{[k]} \in \mathbb{C}^{(m - 2) \times n}$}}
            \nl\For{$i = 3, \ldots, m$}{
                {\setstretch{1.5}
                    \nl $G_i^{[k]} = \frac{1}{2n} Y_{1:2,\bullet}^{[k]} \diag(\Lambda_{i, \bullet}) \, (Y_{1:2,\bullet}^{[k]})^{\adjoint}$\;
                    \nl $\bm{g}_i^{[k]} = \frac{1}{2n} Y_{1:2,\bullet}^{[k]} \diag(\Lambda_{i, \bullet}) \, (Y_{i, \bullet}^{[k]})^{\adjoint}$\;
                    \nl $Y_{i, \bullet}^{[k]} \leftarrow Y_{i, \bullet}^{[k]} - (\bm{g}_i^{[k]})^{\adjoint} (G_i^{[k]})^{-1} Y_{1:2,\bullet}^{[k]}$\;
                }
            }
            \nl\tcc{\textcolor{black}{Update $Y_{1:2,\bullet}^{[k]} \in \mathbb{C}^{2 \times n}$}}
            \nl\For{$i = 1, 2$}{
                \nl $G_i^{[k]} = \frac{1}{2n} Y_{1:2,\bullet}^{[k]} \diag(\Lambda_{i, \bullet}) \; (Y_{1:2,\bullet}^{[k]})^{\adjoint}$\;
            }
            \nl Update $P^{[k]} \in \mathbb{C}^{2 \times 2}$ using \eqref{eq:ip2:1}--\eqref{eq:ip2:2}.\;
            \nl $Y_{1:2,\bullet}^{[k]} \leftarrow P^{[k]} Y_{1:2,\bullet}^{[k]} \in \mathbb{C}^{2 \times n}$\;
        }
        \nl\tcc{\textcolor{black}{Permute rows}}
        \nl $\Lambda \leftarrow \Pi_2 \Lambda$, where $\Pi_2$ is defined as \eqref{eq:permu-mat}.\;
        \nl $Y^{[k]} \leftarrow \Pi_2 Y^{[k]}$ for $k = 1,\ldots,K$.\;
    }
}
\KwOut{$Y^{[k]} \in \mathbb{C}^{m \times n}$ $(k = 1,\ldots,K)$}
}
\end{algorithm}

\section{Relation to Prior MM+BCD Algorithms}

\textit{Iterative projection (\text{IP})} is a family of MM+BCD algorithms that optimize several
\underline{rows of $W$} in each iteration with closed-form update formulas.
So far, $\text{IP}_1$~\cite{ono2010auxica,ono2011auxiva,ablin2019stochastic} and $\text{IP}_2$~\cite{ono2018asj}
(see also~\cite{nakashima2021faster,scheibler2020ive,scheibler2021ipa,ike2021ive}) have been developed as members of $\text{IP}$.

$\text{IP}_1$ is an MM+BCD that updates 
$\Lambda \to \bm{w}_1 \to \cdots \to \bm{w}_m$
one by one.
% (the index $\empty^{[k]}$ is omitted for the sake of simplicity).
The update rule for $\Lambda$ is given by \eqref{eq:lambda}
and that for $\bm{w}_{\ell} \coloneqq W_{\ell, \bullet}^{\adjoint}$ can be developed as
\begin{align}
\nonumber
\bm{u}_\ell &\leftarrow (W V_\ell)^{-1} \bm{e}_\ell \in \mathbb{C}^{m \times 1},
\quad
\bm{w}_\ell \leftarrow \frac{\bm{u}_\ell}{(\bm{u}_\ell^{\adjoint} V_\ell \bm{u}_\ell)^{\frac{1}{2}}} \in \mathbb{C}^{m \times 1},
\end{align}
where $V_{\ell}$ is defined by \eqref{eq:Vi}
and $\bm{e}_{\ell}$ is the $\ell$-th column of $I_m$.

$\text{IP}_2$ is an MM+BCD that updates
$\Lambda \to [\bm{w}_1,\bm{w}_2] \to \cdots \to [\bm{w}_{m - 1}, \bm{w}_m]$
one by one (when $m$ is even), which improves $\text{IP}_1$
(see, e.g.,~\cite{nakashima2021faster,scheibler2020ive,scheibler2021ipa,ike2021ive} for details).
%(Due to space limitations, the derivation is omitted here.)

%\subsubsection{Iterative Projection with Adjustment (IPA) Algorithm}
Recently, an advanced algorithm called \textit{iterative projection with adjustment (IPA)} was proposed~\cite{scheibler2021ipa}.
However, unlike IP and ISS, no (fully) closed-form update formula has existed for IPA,
because it requires a root-finding algorithm (and for this purpose the Newton-Raphson method is used~\cite{scheibler2021ipa}).
Although IPA is important, we will not compare it with IP and ISS in our experiments,
since we are focusing on such methods with fully closed-form update formulas.

\section{Time Complexity Analysis}

The computational time complexity of $\text{ISS}_2$ per MM iteration is dominated by
\begin{itemize}
\item the computation of $(G_i^{[k]}, \bm{g}_i^{[k]})$ for each $i = 1,\ldots,m$ and loop $\ell = 1,\ldots,\frac{m}{2}$,
which costs $\mathrm{O}(K m^2 n)$; and
\item the computation of $Y^{[k]}$, which costs $\mathrm{O}(K m^2 n)$.
\end{itemize}
Thus, $\text{ISS}_2$ has the time complexity of $\mathrm{O}(K m^2 n)$, which is the same as that of $\text{ISS}_1$.
For comparison, the time complexity of $\text{IP}_1$ with $d \in \{1, 2\}$ per iteration is dominated by
(e.g., \cite{scheibler2021ipa})
\begin{itemize}
\item the computation of covariance matrices $V_1^{[k]},\ldots,V_m^{[k]} \in \mathcal{S}_{+}^m$,
which costs $\mathrm{O}(K m^3 n)$; and
\item the computation of updating $W^{[k]}$, which costs $\mathrm{O}(K m^4)$.
\end{itemize}
Thus, $\text{IP}_d$ ($d \in \{1, 2\}$) has the time complexity of $\mathrm{O}(K m^3 n + K m^4)$, which is $m$ times larger than $\text{ISS}_d$ with $d \in \{ 1, 2\}$.

\section{Experiments}
\label{sec:exp}

We compared the performance of our proposed $\text{ISS}_2$ and conventional $\text{ISS}_1$, $\text{IP}_1$, and $\text{IP}_2$ when applied to convolutive blind source separation (BSS) in the short-time Fourier transform (STFT) domain~\cite{vincent2018book},
where $K$ and $n$ correspond to the numbers of frequency bins and time frames, respectively.
This setting is very common in audio source separation~\cite{vincent2018book}.

%\subsection{Experimental conditions}

\textbf{Dataset:}
We generated synthesized convolutive mixtures of $m \in \{ 4, 6, 8, 10 \}$ speech signals.
The signals were captured by a circular array with $m$ microphones and a radius of 5 cm.
We obtained speech signals from the TIMIT corpus~\cite{timit}
and concatenated them so that the signal length exceeded 10 seconds.
The obtained signals were normalized to have unit power.
To obtain acoustic impulse responses (AIR),
we used the \verb|pyroomacoustics| Python package~\cite{scheibler2018pyroomacoustics}
and simulated 100 rectangular rooms.
The rooms were 5 to 8 m wide and 3~to 5~m high.
The arrays were placed in the center of the rooms at a height of 1~m.
The speech sources were randomly placed in the room at a height of 1~m,
provided that the distances from the array center and the walls were at least 1~m.
The reverberation times ($\text{T}_{60}$) ranged from 250 to 400~ms.

\textbf{Evaluation criterion:}
We measured the signal-to-distortion ratio (SDR)~\cite{vincent2006sdr}
between separated signal $\hat{\bm{s}}$ and oracle reverberant speech signal $\bm{s}$ at the first microphone.
The SDR we used here is sometimes called the scale-invariant SDR~\cite{le2019sdr} and defined as
$\text{SDR [dB]} = 10 \log_{10} \frac{ \| \alpha \bm{s} \|_2^2 }{ \| \bm{\hat{s}} - \alpha \bm{s} \|_2^2 }$
with
$\alpha = \frac{\bm{\hat{s}}^\top \bm{s}}{ \| \bm{s} \|_2^2 }$.

\textbf{Other conditions:}
The sampling rate was 16 kHz,
the STFT frame size was 4096 (256 ms),
and the frame shift was 1024 (64 ms).
We assumed a Laplace distribution, i.e., \eqref{eq:GGD} with $\beta = 1$, for the separated signals.
We initialized $W^{[k]}$ as the whitening matrix
$D^{-1/2} U^{\adjoint}$ using the eigenvalue decomposition
$U D U^{\adjoint} = \frac{1}{n} X^{[k]} (X^{[k]})^{\adjoint}$ for each $k = 1,\ldots,K$.
After separation,
the scale ambiguity of IVA, i.e., \eqref{eq:ambiguity}, was restored based on
the minimum distortion principle (MDP)~\cite{matsuoka2001projection-back}
(see also~\cite[Section 2.2]{scheibler2020projection-back} for the details of MDP).

\textbf{Experimental results:}
Figure~\ref{fig:exp-res} shows the SDR improvement obtained by each method.
As we desired,
the convergence of the proposed $\text{ISS}_2$ is much faster than $\text{ISS}_1$ and $\text{IP}_1$
and comparable to $\text{IP}_2$
(note that the SDR curves of $\text{IP}_2$ and $\text{ISS}_2$ almost overlap),
which clearly shows the effectiveness of our approach.
Since the time complexity of $\text{ISS}_2$ is $m$ times smaller than $\text{IP}_2$,
one might expect that
the runtime of $\text{ISS}_2$ to reach convergence is shorter than that of $\text{IP}_2$;
but this was not the case in our experiment with our Python implementation where
the runtime of $\text{ISS}_2$ was slightly inferior to that of $\text{IP}_2$.
This implementation issue is an important future work.

\begin{figure}[!t]
\centering
\begin{subfigure}{0.49\linewidth}
\includegraphics[width=1.0\linewidth]{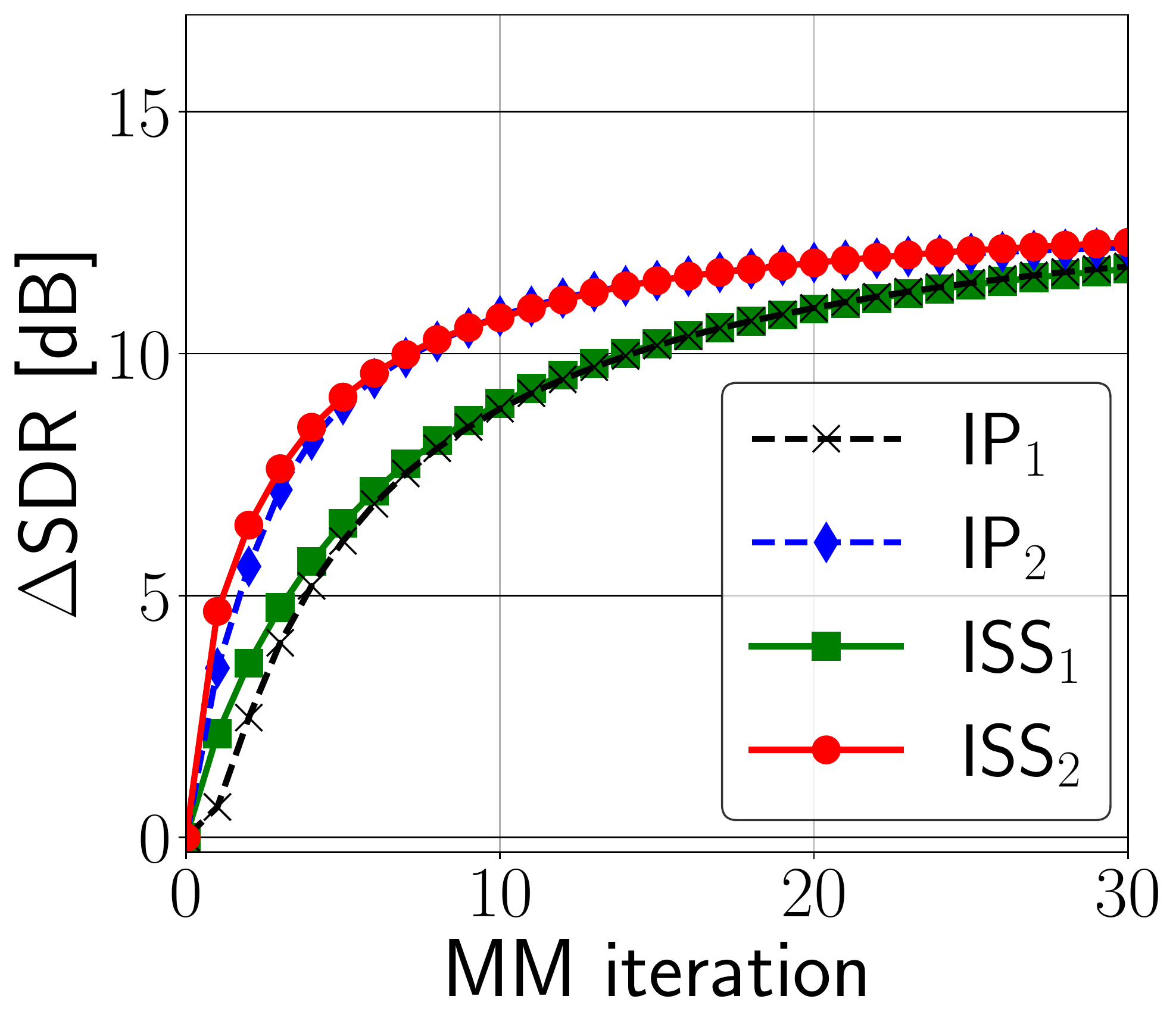}

\vspace{-2 mm}
\caption{$m = 4$}
\end{subfigure}
\hfill
\begin{subfigure}{0.49\linewidth}
\includegraphics[width=1.0\linewidth]{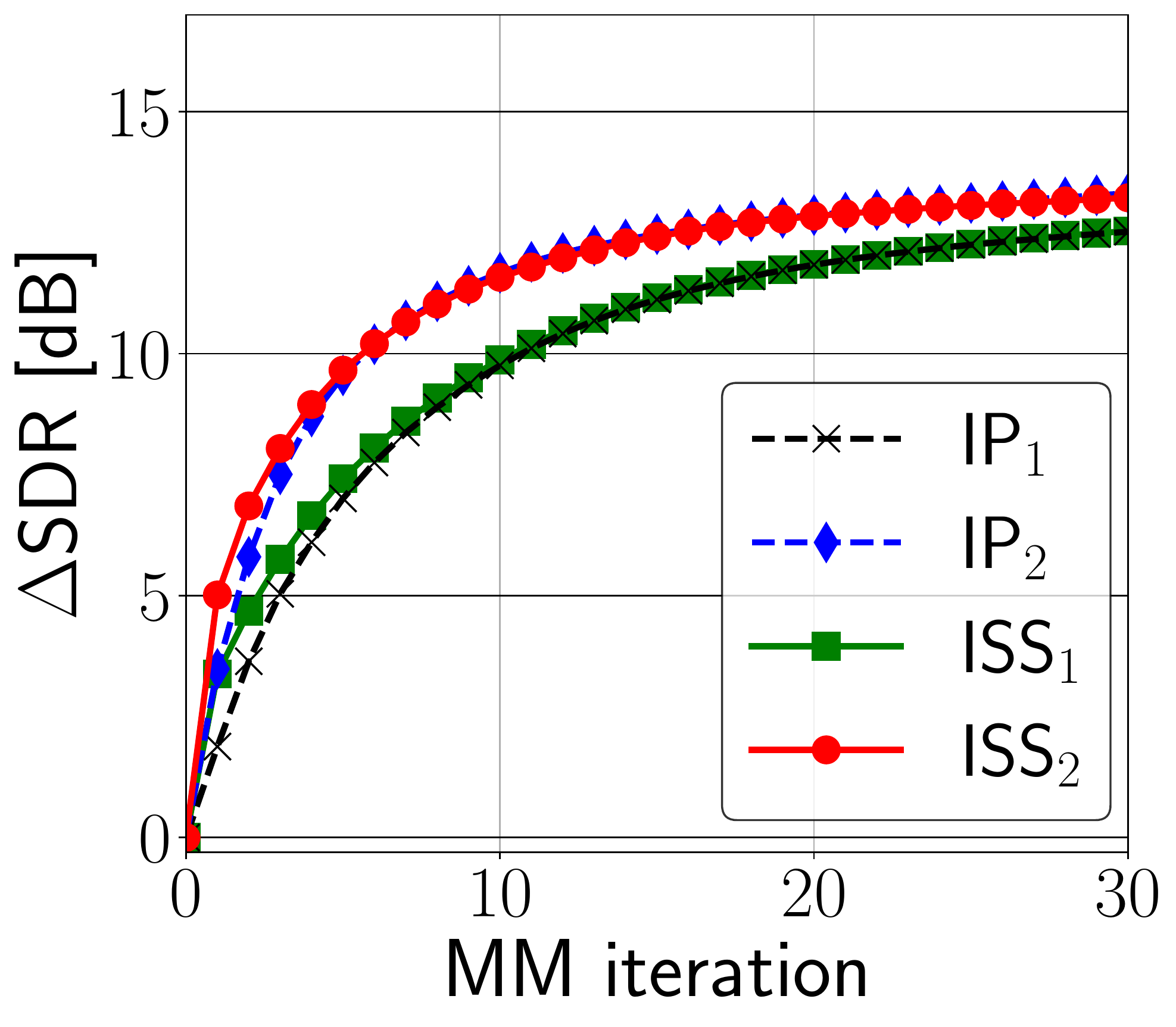}

\vspace{-2 mm}
\caption{$m = 6$}
\end{subfigure}
\vspace{1 mm}
%\hfill

%
\begin{subfigure}{0.49\linewidth}
\includegraphics[width=1.0\linewidth]{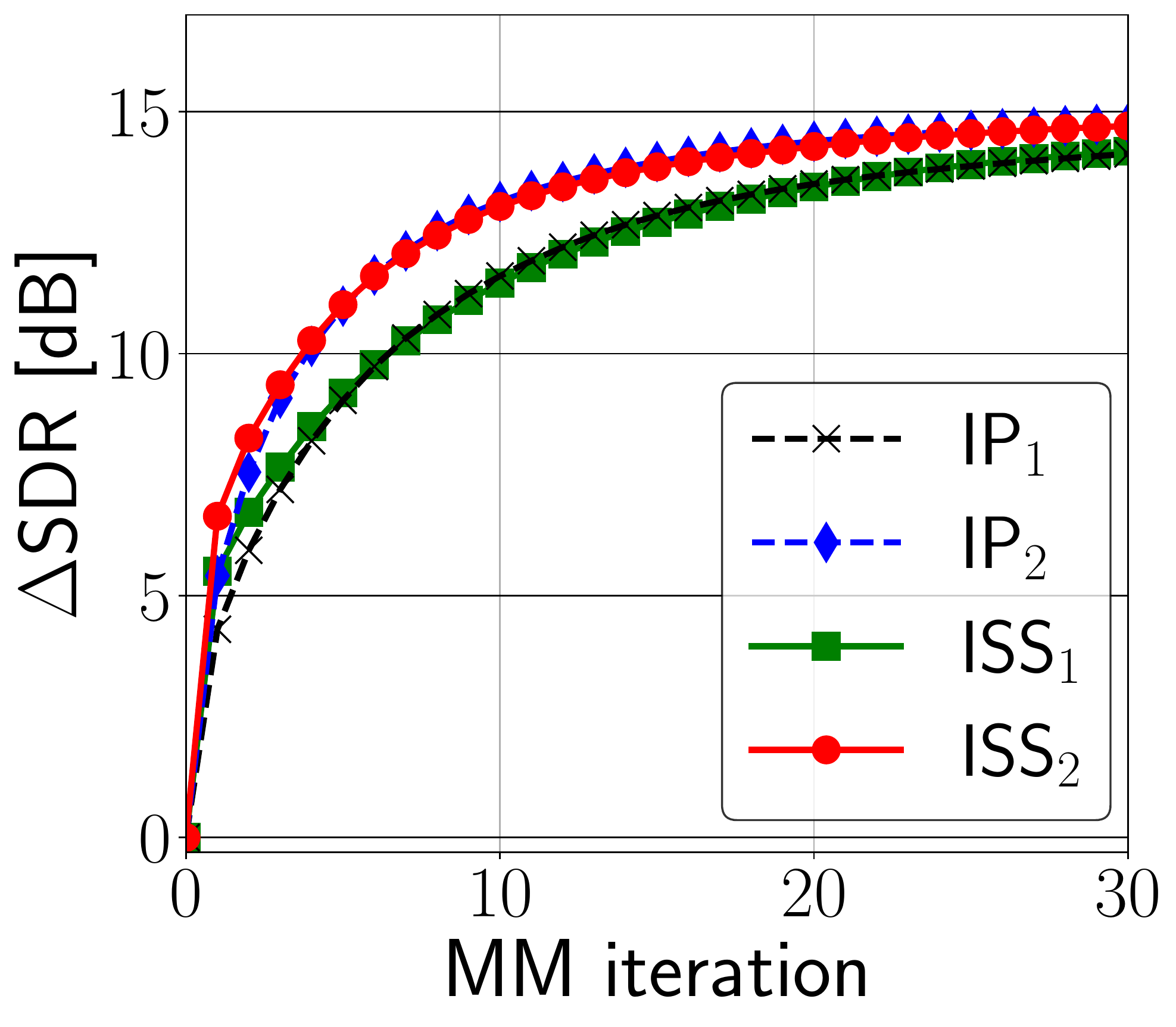}

\vspace{-2 mm}
\caption{$m = 8$}
\end{subfigure}
\hfill
\begin{subfigure}{0.49\linewidth}
\includegraphics[width=1.0\linewidth]{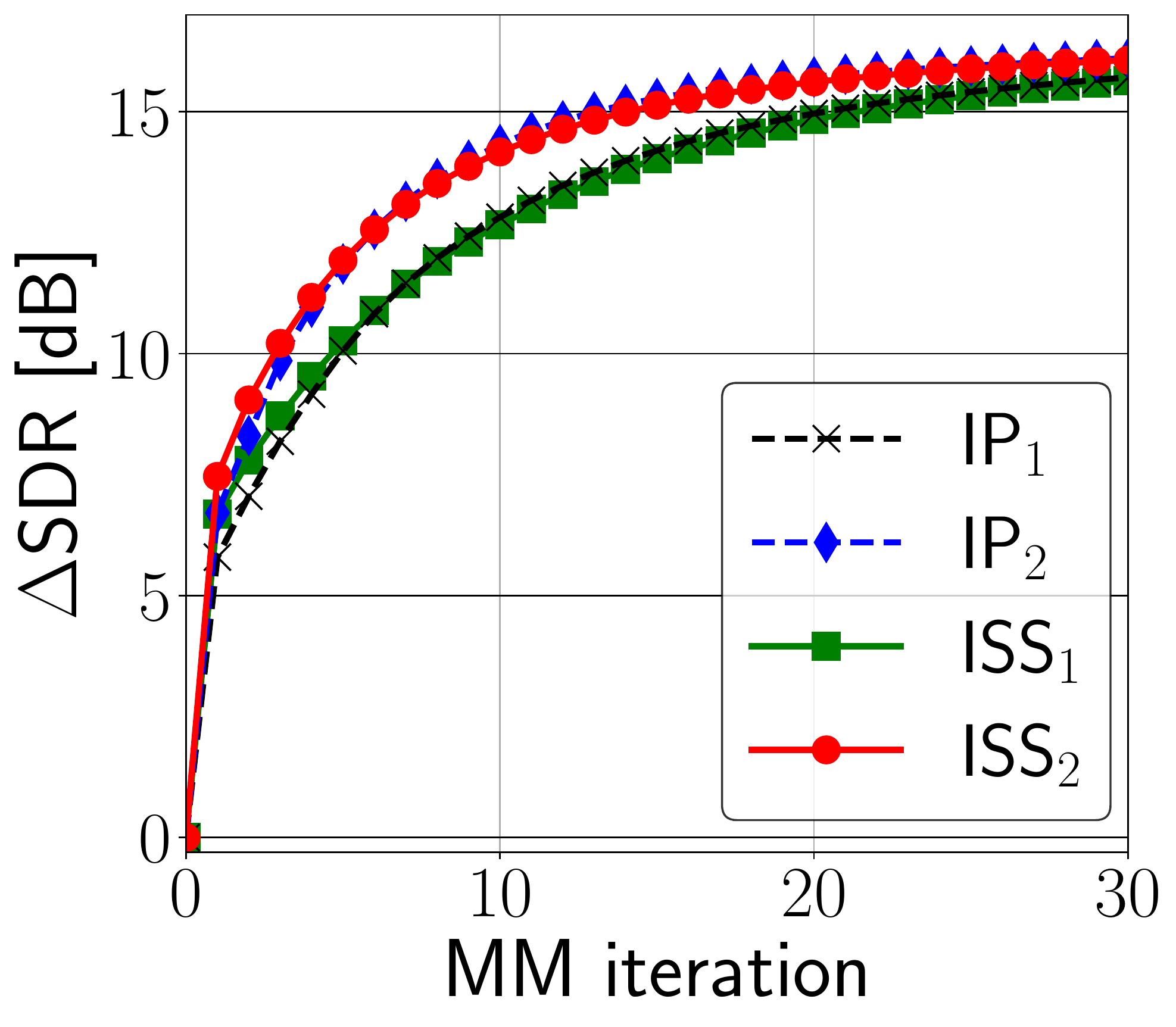}

\vspace{-2 mm}
\caption{$m = 10$}
\end{subfigure}
\caption{
The SDR improvement ($\Delta\text{SDR}$) from the initial SDR as a function of the MM iteration.
The SDRs were averaged over 100 samples.
The average signal length was 13.3 sec.
}
\label{fig:exp-res}
\vspace{-3 mm}
\end{figure}

\section{Conclusion}

As BCD algorithms for the MM-based IVA, $\text{IP}_1$, $\text{IP}_2$, and $\text{ISS}_1$ had been developed.
We here extended $\text{ISS}_1$ to $\text{ISS}_2$ that updates
two columns of the mixing matrix $A \coloneqq W^{-1}$ in each BCD iteration.
Our $\text{ISS}_2$ simultaneously achieves both
(i) the small time complexity of $\text{ISS}_1$ per MM iteration
and
(ii) the fast convergence behavior of $\text{IP}_2$, which was confirmed by the numerical experiments.

%\section*{Acknowledgment}
\bibliographystyle{IEEEtran}
\bibliography{refs}

\begin{thebibliography}{10}
\providecommand{\url}[1]{#1}
\csname url@samestyle\endcsname
\providecommand{\newblock}{\relax}
\providecommand{\bibinfo}[2]{#2}
\providecommand{\BIBentrySTDinterwordspacing}{\spaceskip=0pt\relax}
\providecommand{\BIBentryALTinterwordstretchfactor}{4}
\providecommand{\BIBentryALTinterwordspacing}{\spaceskip=\fontdimen2\font plus
\BIBentryALTinterwordstretchfactor\fontdimen3\font minus
  \fontdimen4\font\relax}
\providecommand{\BIBforeignlanguage}[2]{{%
\expandafter\ifx\csname l@#1\endcsname\relax
\typeout{** WARNING: IEEEtran.bst: No hyphenation pattern has been}%
\typeout{** loaded for the language `#1'. Using the pattern for}%
\typeout{** the default language instead.}%
\else
\language=\csname l@#1\endcsname
\fi
#2}}
\providecommand{\BIBdecl}{\relax}
\BIBdecl

\bibitem{comon2010handbook}
P.~Comon and C.~Jutten, \emph{Handbook of Blind Source Separation:
  {Independent} component analysis and applications}.\hskip 1em plus 0.5em
  minus 0.4em\relax Academic press, 2010.

\bibitem{kim2007}
T.~Kim, H.~T. Attias, S.-Y. Lee, and T.-W. Lee, ``Blind source separation
  exploiting higher-order frequency dependencies,'' \emph{IEEE Trans. Audio,
  Speech, Language Process.}, vol.~15, no.~1, pp. 70--79, 2007.

\bibitem{afsari2008sensitivity}
B.~Afsari, ``Sensitivity analysis for the problem of matrix joint
  diagonalization,'' \emph{SIAM Journal on Matrix Analysis and Applications},
  vol.~30, no.~3, pp. 1148--1171, 2008.

\bibitem{anderson2014independent}
M.~Anderson, G.-S. Fu, R.~Phlypo, and T.~Adal{\i}, ``Independent vector
  analysis: {Identification} conditions and performance bounds,'' \emph{IEEE
  Trans. Signal Process.}, vol.~62, no.~17, pp. 4399--4410, 2014.

\bibitem{lahat2018joint}
D.~Lahat and C.~Jutten, ``Joint independent subspace analysis: {Uniqueness} and
  identifiability,'' \emph{IEEE Trans. Signal Process.}, vol.~67, no.~3, pp.
  684--699, 2018.

\bibitem{lahat2016alternative}
------, ``An alternative proof for the identifiability of independent vector
  analysis using second order statistics,'' in \emph{Proc. ICASSP}, 2016, pp.
  4363--4367.

\bibitem{bell1995infomax}
A.~J. Bell and T.~J. Sejnowski, ``An information-maximization approach to blind
  separation and blind deconvolution,'' \emph{Neural computation}, vol.~7,
  no.~6, pp. 1129--1159, 1995.

\bibitem{cardoso1996equivariant}
J.-F. Cardoso and B.~H. Laheld, ``Equivariant adaptive source separation,''
  \emph{IEEE Trans. Signal Process.}, vol.~44, no.~12, pp. 3017--3030, 1996.

\bibitem{amari1996natural-gradient}
S.~Amari, A.~Cichocki, and H.~H. Yang, ``A new learning algorithm for blind
  signal separation,'' in \emph{Proc. NIPS}, 1996, pp. 757--763.

\bibitem{zibulevsky2003blind}
M.~Zibulevsky, ``Blind source separation with relative newton method,'' in
  \emph{Proc. ICA}, 2003, pp. 897--902.

\bibitem{palmer2008newton}
J.~A. Palmer, S.~Makeig, K.~Kreutz-Delgado, and B.~D. Rao, ``Newton method for
  the {ICA} mixture model,'' in \emph{Proc. ICASSP}, 2008, pp. 1805--1808.

\bibitem{choi2007relative}
H.~Choi and S.~Choi, ``A relative trust-region algorithm for independent
  component analysis,'' \emph{Neurocomputing}, vol.~70, no. 7-9, pp.
  1502--1510, 2007.

\bibitem{ablin2018faster}
P.~Ablin, J.-F. Cardoso, and A.~Gramfort, ``Faster independent component
  analysis by preconditioning with {Hessian} approximations,'' \emph{IEEE
  Trans. Signal Process.}, vol.~66, no.~15, pp. 4040--4049, 2018.

\bibitem{komodakis2015playing}
N.~Komodakis and J.-C. Pesquet, ``Playing with duality: {An} overview of recent
  primal-dual approaches for solving large-scale optimization problems,''
  \emph{IEEE Signal Processing Magazine}, vol.~32, no.~6, pp. 31--54, 2015.

\bibitem{yatabe2018determined}
K.~Yatabe and D.~Kitamura, ``Determined blind source separation via proximal
  splitting algorithm,'' in \emph{Proc. ICASSP}, 2018, pp. 776--780.

\bibitem{yatabe2021determined}
------, ``Determined {BSS} based on time-frequency masking and its application
  to harmonic vector analysis,'' \emph{IEEE/ACM Trans. Audio, Speech, Language
  Process.}, vol.~29, pp. 1609--1625, 2021.

\bibitem{hyvarinen1999fastica}
A.~Hyvarinen, ``Fast and robust fixed-point algorithms for independent
  component analysis,'' \emph{IEEE Trans. Neural Netw.}, vol.~10, no.~3, pp.
  626--634, 1999.

\bibitem{ablin2018faster-O}
P.~Ablin, J.-F. Cardoso, and A.~Gramfort, ``Faster {ICA} under orthogonal
  constraint,'' in \emph{Proc. ICASSP}, 2018, pp. 4464--4468.

\bibitem{lange2016mm}
K.~Lange, \emph{MM optimization algorithms}.\hskip 1em plus 0.5em minus
  0.4em\relax SIAM, 2016.

\bibitem{ono2010auxica}
N.~Ono and S.~Miyabe, ``Auxiliary-function-based independent component analysis
  for super-{Gaussian} sources,'' in \emph{Proc. LVA/ICA}, 2010, pp. 165--172.

\bibitem{ablin2019stochastic}
P.~Ablin, A.~Gramfort, J.-F. Cardoso, and F.~Bach, ``Stochastic algorithms with
  descent guarantees for {ICA},'' in \emph{Proc. AISTATS}, 2019, pp.
  1564--1573.

\bibitem{ono2011auxiva}
N.~Ono, ``Stable and fast update rules for independent vector analysis based on
  auxiliary function technique,'' in \emph{Proc. WASPAA}, 2011, pp. 189--192.

\bibitem{pham2001}
D.-T. Pham and J.-F. Cardoso, ``Blind separation of instantaneous mixtures of
  nonstationary sources,'' \emph{IEEE Trans. Signal Process.}, vol.~49, no.~9,
  pp. 1837--1848, 2001.

\bibitem{degerine2007maxdet}
S.~D{\'e}gerine and A.~Za{\"\i}di, ``Determinant maximization of a nonsymmetric
  matrix with quadratic constraints,'' \emph{SIAM J. Optim.}, vol.~17, no.~4,
  pp. 997--1014, 2007.

\bibitem{yeredor2012SeDJoCo}
A.~Yeredor, B.~Song, F.~Roemer, and M.~Haardt, ``A “sequentially drilled”
  joint congruence {(SeDJoCo)} transformation with applications in blind source
  separation and multiuser {MIMO} systems,'' \emph{IEEE Trans. Signal
  Process.}, vol.~60, no.~6, pp. 2744--2757, 2012.

\bibitem{nocedal-Jorge2006optimization}
J.~Nocedal and S.~Wright, \emph{Numerical optimization}.\hskip 1em plus 0.5em
  minus 0.4em\relax Springer Science \& Business Media, 2006.

\bibitem{ono2018asj}
N.~Ono, ``Fast algorithm for independent component/vector/low-rank matrix
  analysis with three or more sources,'' in \emph{Proc. ASJ Spring Meeting},
  2018, (in Japanese).

\bibitem{nakashima2021faster}
T.~Nakashima, R.~Scheibler, Y.~Wakabayashi, and N.~Ono, ``Faster independent
  low-rank matrix analysis with pairwise updates of demixing vectors,'' in
  \emph{Proc. EUSIPCO}, 2021, pp. 301--305.

\bibitem{scheibler2020ive}
R.~Scheibler and N.~Ono, ``{MM} algorithms for joint independent subspace
  analysis with application to blind single and multi-source extraction,''
  \emph{arXiv:2004.03926v1}, 2020.

\bibitem{scheibler2021ipa}
R.~Scheibler, ``Independent vector analysis via log-quadratically penalized
  quadratic minimization,'' \emph{IEEE Trans. Signal Process.}, vol.~69, pp.
  2509--2524, 2021.

\bibitem{ike2021ive}
R.~Ikeshita, T.~Nakatani, and S.~Araki, ``Block coordinate descent algorithms
  for auxiliary-function-based independent vector extraction,'' \emph{IEEE
  Trans. Signal Process.}, vol.~69, pp. 3252--3267, 2021.

\bibitem{scheibler2020iss}
R.~Scheibler and N.~Ono, ``Fast and stable blind source separation with rank-1
  updates,'' in \emph{Proc. ICASSP}, 2020, pp. 236--240.

\bibitem{palmer2006variational}
J.~Palmer, D.~Wipf, K.~Kreutz-Delgado, and B.~Rao, ``Variational {EM}
  algorithms for {non-Gaussian} latent variable models,'' in \emph{Proc. NIPS},
  vol.~18, 2005, pp. 1059--1066.

\bibitem{benveniste1990}
A.~Benveniste, M.~M{\'e}tivier, and P.~Priouret, \emph{Adaptive algorithms and
  stochastic approximations}, 1st~ed.\hskip 1em plus 0.5em minus 0.4em\relax
  Springer Science, 1990, vol.~22.

\bibitem{ono2012auxiva-stereo}
N.~Ono, ``Fast stereo independent vector analysis and its implementation on
  mobile phone,'' in \emph{Proc. IWAENC}, 2012, pp. 1--4.

\bibitem{vincent2018book}
E.~Vincent, T.~Virtanen, and S.~Gannot, \emph{Audio source separation and
  speech enhancement}.\hskip 1em plus 0.5em minus 0.4em\relax John Wiley \&
  Sons, 2018.

\bibitem{timit}
J.~Garofolo, L.~Lamel, W.~Fisher, J.~Fiscus, D.~Pallett, N.~Dahlgren, and
  V.~Zue, ``{TIMIT Acoustic-Phonetic Continuous Speech Corpus LDC93S1},'' Web
  Download. Philadelphia: Linguistic Data Consortium, Tech. Rep., 1993.

\bibitem{scheibler2018pyroomacoustics}
R.~Scheibler, E.~Bezzam, and I.~Dokmani{\'c}, ``Pyroomacoustics: {A} {Python}
  package for audio room simulation and array processing algorithms,'' in
  \emph{Proc. ICASSP}, 2018, pp. 351--355.

\bibitem{vincent2006sdr}
E.~Vincent, R.~Gribonval, and C.~F{\'e}votte, ``Performance measurement in
  blind audio source separation,'' \emph{IEEE Trans. Audio, Speech, Language
  Process.}, vol.~14, no.~4, pp. 1462--1469, 2006.

\bibitem{le2019sdr}
J.~Le~Roux, S.~Wisdom, H.~Erdogan, and J.~R. Hershey, ``{SDR} -- half-baked or
  well done?'' in \emph{Proc. ICASSP}, 2019, pp. 626--630.

\bibitem{matsuoka2001projection-back}
K.~Matsuoka and S.~Nakashima, ``Minimal distortion principle for blind source
  separation,'' in \emph{Proc. ICA}, 2001, pp. 722--727.

\bibitem{scheibler2020projection-back}
R.~Scheibler, ``Generalized minimal distortion principle for blind source
  separation,'' in \emph{Proc. Interspeech}, 2020, pp. 3326--3330.

\end{thebibliography}
\end{document}